\documentclass[runningheads]{llncs}
\usepackage[english]{babel}
\usepackage{tabularx}
\usepackage{textcomp}
\usepackage{xcolor}
\usepackage{booktabs}
\usepackage{mathrsfs}
\usepackage{hyperref}
\usepackage{multicol}

\usepackage[utf8]{inputenc}
\usepackage[T1]{fontenc}
\usepackage{amsmath}
\usepackage{amssymb}
\usepackage{graphicx}
\usepackage{algorithm}
\usepackage{algpseudocode}
\begin{document}
\title{Approximation Algorithms for Drone Delivery Scheduling Problem}

%

\author{Saswata Jana  \and
Partha Sarathi Mandal\orcidID{0000-0002-8632-5767} 
}
\authorrunning{Jana et al.}
%
\institute{Indian Institute of Technology Guwahati, Guwahati - 781039, India\\
\email{\{saswatajana, psm\}@iitg.ac.in}\\
}
\maketitle              
\begin{abstract}
The coordination among drones and ground vehicles for last-mile delivery 
has gained significant interest in recent years.
In this paper, we study \textit{multiple drone delivery scheduling problem} (MDSP) \cite{Betti_ICDCN22}  for last-
mile delivery, where we have a set of drones with an identical battery budget and a set of delivery locations, along  
with reward or profit for delivery, cost and delivery time intervals. The objective of the MDSP is to find a collection of conflict-free schedules for each drone such that the total profit for delivery is maximum subject to the battery constraint of the drones. Here we propose a fully polynomial time approximation scheme (FPTAS) for the single drone delivery scheduling problem (SDSP) and a $\frac{1}{4}$-approximation algorithm for MDSP with a constraint on the number of drones.
\end{abstract}

\section{Introduction}

\textit{Motivation}:
The rapid demand for commercial deliveries motivates logistics to provide more efficient service to their customers. In delivery, the \textit{last-mile delivery} \cite{LastMile}, which involves collecting the goods from the distribution center to the customer's door, is the most costly and time-consuming part of this procedure. This type of delivery necessitates a significant amount of human interactions. However, advances in drone technology have created a miniature, and this delivery system has a more significant impact in today's pandemic world. Large corporations have started planning for efficient parcel delivery via drones \cite{amazon}. Also, drones can bypass traffic congestion on traditional networks by flying. This ability enables the drone to travel faster than the usual delivery truck and reduces transportation costs. Drones are carried by truck in the last mile delivery for better efficiency. In this collaboration, drones take off from the truck to deliver the package to customers and then return to the truck for the next delivery. In addition drones have enormous applications in the field of agriculture \cite{agriculture}, healthcare \cite{healthcare}, defense and disaster response \cite{defence}, resource monitoring and assessment \cite{Sibanda}, manufacturing \cite{Maghazei}, and many more.

\noindent \textit{Challenges:}
In spite of the broad application of drones, it has certain limitations. A drone can only deliver small packages due to capacity constraint. Since customers' locations are geographically distributed, safety and reliability are critical issues. Package delivery by drones with the help of a truck has numerous challenges. For a given set of customer locations, we need to know the best route for the truck, the launching points for drones and rendezvous points for drones and the truck to make deliveries.
Furthermore, the limited battery budget of the available drones and the limited number of drones in the company's warehouse prevent us from making a desirable number of deliveries. Various companies have introduced preferences among customers (like prime customers). For this reason, they want to maximize their profit by preferring those prime deliveries over others. At a time, a drone can deliver at most one package. All these complexities influence logistics to make those deliveries, for which the profit is maximized. 


\noindent \textit{Drone delivery
scheduling problem:}
This paper considers last-mile delivery to customers with a truck moving in a prescribed route and a set of drones. For a given set of deliveries and their delivery time intervals, reward or profit for each delivery, and battery budget of the drones, the goal is to schedule the drones for the deliveries to make the total profit maximum. This problem was introduced by Sorbelli et al. \cite{Betti_ICDCN22} and proved that the problem is NP-Hard and proposed heuristic algorithms to solve the problems. In this paper, we offer two approximation algorithms as mentioned below.   
\subsection{Our Contributions}
In this paper, our contributions are the following.
\begin{itemize}
\item For the \textit{single drone delivery scheduling problem}, we propose an FPTAS with running time $\mathcal{O}$($n^{2}  \lfloor \frac{n}{\epsilon} \rfloor )$, for $\epsilon > 0$, where $n$ is the number of deliveries.
\item For the \textit{multiple drone delivery scheduling problem} with m drones, we propose an $\frac{1}{4}$-approximation algorithm with running time $\mathcal{O}$($n^{2}$) if $m \geq \Delta$, where $\Delta$ is the maximum degree of the interval graph constructed from the delivery intervals.
\end{itemize}

\subsection{Related Works}
Due to the drone's limited flying endurance capacity, collaboration with a truck allows for more effective customer deliveries. Various studies have been conducted in this area of co-operative delivery.

This type of delivery is taken into account when \textit{Murray and Chu} in \cite{MURRAY201586} first introduced \textit{flying sidekicks traveling salesman problem}, a more extension of the traveling salesman problem (TSP). In this model, drones can launch from the depot or any customer location to deliver the package. Then it returns to the truck at any customer location or the depot. In addition, the truck can also deliver the packages and make some synchrony for launching and receiving the drone. But, any customer can be served by either the truck or a drone. This paper aims to reduce the total time to make all the deliveries. They proposed an optimal \textit{mixed integer linear programming} (MILP) formulation and two effective heuristic solutions with numerical testing for this purpose.
Furthermore, \textit{Murray and Raj} in \cite{Murray2019TheMF} extended the above problem for multiple drone cases, referred to as \textit{multiple flying sidekicks traveling salesman problem}. Here, the authors offered MILP formulation for the problem and a heuristic solution with simulation. Using the solution for the TSP, \textit{Crisan and Nechita} in  \cite{CRISAN201938} suggested another competent heuristic for the \textit{flying sidekicks traveling salesman problem}. 

The mobility of drones came into the literature when \textit{Daknama and Kraus} in \cite{Daknama2017VehicleRW} took a policy for recharging the drones on the truck's roof. With this, a drone can deliver more packages than the usual method. In addition, a heuristic algorithm has been proposed by the authors for the schedule of the truck and drones. 
\textit{Mara et al.} in \cite{Mara} formulate flying sidekicks salesman problem with multiple drops. In this model, a drone can deliver multiple packages in a single sortie. Here they proposed a heuristic approach with effective testings for the problem.
\textit{Boysen et al.} in \cite{Boysen2018DroneDF} consider the delivery by drones only. Their objective is to find launch and rendezvous points of the truck and a drone for delivery such that the total makespan to complete all the deliveries becomes minimum. 
Also, \textit{Sorbelli et al.} in \cite{Betti_ICDCN22} proposed a model for delivery by drones only. Here they offered a Multiple Drone-Delivery Scheduling Problem (MDSP), where a truck and multiple drones cooperate among them-self to deliver a package in the last mile. The objective of this problem is to find the schedule for the available drones to maximize profit. 
The authors propose NP-hardness proof of the problem, ILP formulation, a heuristic algorithm for the single drone case, and two heuristic algorithms for the multiple drones case.
In this paper, we study these problems and propose two approximation algorithms; one for the single drone and the other for multiple drones.

The rest of the paper is organized as follows.
The model with problem formulation is presented in Section \ref{section:Model}. Approximation algorithm for single drone is proposed in Section \ref{section:singleDrone} and for multiple drone is presented in Section \ref{section:multiDrone}. Finally, we conclude the paper in Section \ref{section:conclusion}.
\vspace{-2mm}
\section{Model and Problem Formulation}
\label{section:Model}
\vspace{-3mm}
In this section, we describe the model of the Multiple Drone-Delivery Scheduling Problem (MDSP) and integer linear programming (ILP) formulation \cite{Betti_ICDCN22} of it. 
A table \ref{tab:notation} gives a list of variables with descriptions used in this paper at the end.

\vspace{-4mm}
\subsection{Model} Let $\mathcal{M} = \{1, 2, \cdots, m\}$ be the set of drones, each have equal battery budget $B$ and $\mathcal{N} = \{1, 2,\cdots, n\}$ be the set of customers or deliveries with prescribed location $\delta_{j}$ for each $j$ $\in$ $\mathcal{N}$. For each customer $j$, there are associate \textit{cost} $c_{j}$ and \textit{profit} $p_{j}$, which are incurred by a drone for completing the delivery $j$.    
Initially, all the drones are at the depot. 
A truck starts it's journey containing all the drones on a prescribed path for delivering customers.
To complete the delivery to the customer $j$ at the position $\delta_j$, a drone takes off from the truck at the given launch point ($\delta_j^L$), and after completing the delivery, it meets the truck again at the given rendezvous point ($\delta_j^R$).

Let at time $t_{0}$ the truck starts it's journey and $t_{j}^{L}$, $t_{j}^{R}$ be the time when it comes to the points $\delta_{j}^{L}$, $\delta_{j}^{R}$, respectively. So, $I_{j}$ = [$t_{j}^{L}$, $t_{j}^{R}$] is the \textit{delivery time interval} for the delivery $j$. Any delivery $j$ can be performed by at most one drone, and that drone cannot be assigned to any other delivery at the time interval $I_{j}$. Also, note that the truck cannot move back at any time, i.e., if $P$ and $Q$ are two points on the path of the truck, where $Q$ is located after $P$, then $t_{P}<t_{Q}$, where $t_{P}$ and $t_{Q}$ are the times when the truck passes through the points $P$ and $Q$, respectively.
Figure \ref{figure-1} shows an example of a drone-delivery model with eight delivery locations and two drones. The solid lines in the figure represent the paths of the truck, while the dotted lines represent the paths of the drone. 
\vspace{-8mm}
\begin{figure}[htp]
    \centering
    \includegraphics[width=10cm]{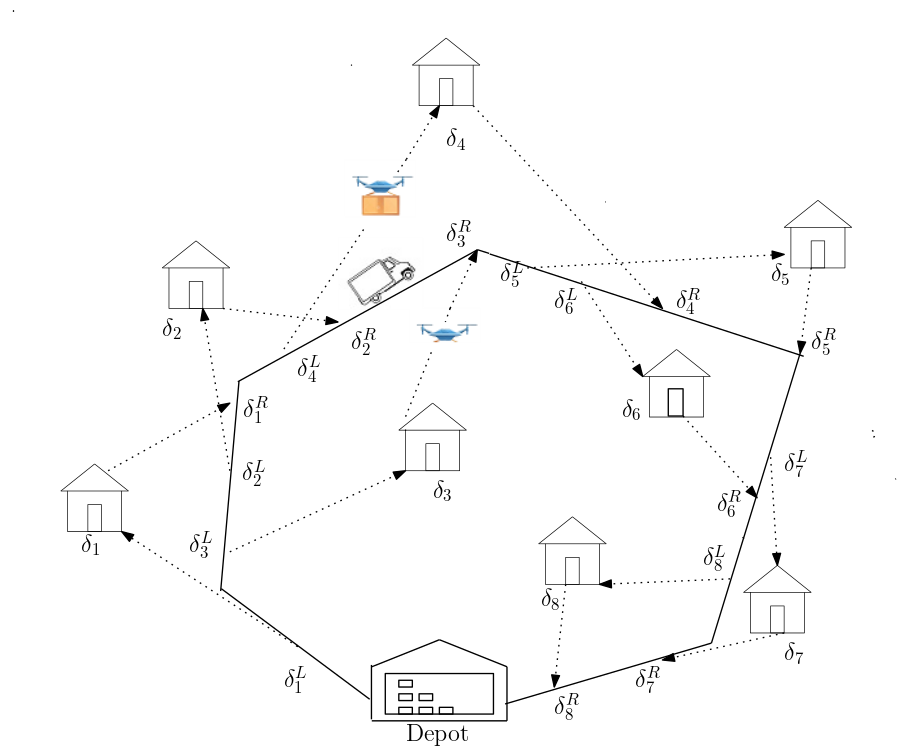}
    \caption{A drone-delivery model with paths specification.}
    \label{figure-1}
\end{figure}

Any two intervals $I_{j}$ and $I_{k}$ are said to be \textit{compatible} if $I_{j} \cap I_{k} = \emptyset$; else they are \textit{conflicting}, where $1\leq j \neq k \leq n$. Let $I = \{I_{1},I_{2}, \cdots, I_{n}\}$ be the set of delivery time intervals corresponding to the set of deliveries $\mathcal{N}$. Without loss of generality we assume that intervals $I_{1},I_{2}, \cdots, I_{n}$ are sorted in non-decreasing order according to their launching time, i.e., $t_{1}^{L}$ $\leq$ $t_{2}^{L}$ $\leq$, $\cdots$, $\leq$ $t_{n}^{L}$.
A set $S \subseteq I$ is said to be \textit{compatible} if all pair of intervals in it are \textit{compatible}. A compatible set $S$ is said to be \textit{feasible}, if the energy cost of $S$, which is defined as $\mathcal{W}$(S) = $\sum_{I_{j} \in S}^{}$ $c_{j} \leq B$. One drone $i \in \mathcal{M}$ can be utilized for multiple deliveries ($S$), provided $S$ is feasible and corresponding profit is $\mathcal{P}$(S) = $\sum_{I_{j} \in S}^{}$ $p_{j}$.
\begin{definition}
An \textit{assignment} $S$ of the deliveries  for a drone is defined as $S = \{I_{k_1}, I_{k_2}$, $ \cdots, I_{k_l}\} \subseteq I$, where $1 \leq k_i \leq n$.
\end{definition}

We are formally rewriting the Multiple Drone-Delivery Scheduling Problem definition \cite{Betti_ICDCN22} as follows.

\begin{problem}
\label{Problem 1}
(Multiple Drone-Delivery Scheduling Problem (MDSP))	\\
Given a set of drones $\mathcal{M}$ with identical battery budget $B$ and interval set $I$ corresponding to a set of deliveries $\mathcal{N}$, the objective is to find a family of assignment $\mathcal{S}^{*}$ = \{$S_{1}^{*}, S_{2}^{*}, \cdots, S_{m}^{*}$ \} corresponding to  $m$ drones such that $\sum_{i=1}^{m} \mathcal{P}(S_{i}^{*}$) is maximized, where $S_{i}^{*} \subseteq I$ is the feasible assignment corresponding to the drone $i \in \mathcal{M}$ and $S_{k}^{*} \cap S_{l}^{*} = \emptyset$; $\forall k, l: 1 \leq k \neq l \leq m$.
\end{problem}

\begin{definition}
An optimum profit function $f$ for $m$ drones is defined as f(m) = $\sum_{i=1}^{m} \mathcal{P}(S_{i}^{*}$).
\end{definition}
	
\subsection{ILP Formulation}
The MDSP is an NP-hard \cite{Betti_ICDCN22} problem, even for the single drone case. For the optimal algorithm, the problem is defined via ILP formulation according to \cite{Betti_ICDCN22} as follows:
Let,
	\begin{equation}
		\begin{split}
			x_{ij} & = 1 \text{, if delivery $j$ $\in$ $\mathcal{N}$ completed by the drone $i$ $\in$ $\mathcal{M}$} \\
			& = 0 \text{, otherwise}.
		\end{split}
	\end{equation}
	\vspace{-4mm}
\begin{equation}
	\max \sum_{i \in \mathcal{M}}^{} \sum_{j \in \mathcal{N}}^{}  p_{j} x_{ij} \label{eq:2}
\end{equation}
subject to
\begin{equation}
	\sum_{j \in \mathcal{N}}^{}  c_{j} x_{ij} \leq B, \text{    \hspace{80pt}     $\forall$ $i$ $\in$ $\mathcal{M}$} \label{eq:3}
	\end{equation}
	\vspace{-4mm}
\begin{equation}
	\sum_{i \in \mathcal{M}}^{} x_{ij} \leq 1, \text{ \hspace{80pt}   $\forall$ j $\in$ $\mathcal{N}$} \label{eq:4}
\end{equation}
\vspace{-4mm}
\begin{equation}
	x_{ij} + x_{ik} \leq 1, \text{  $\forall i \in \mathcal{M}$; $\forall j, k \in \mathcal{N}$: $I_{j} \cap I_{k} \neq \emptyset$} \label{eq:5}
\end{equation}
\vspace{-4mm}
\begin{equation}
	x_{ij} \in \{0,1\}, \text{ \hspace{60pt}    $\forall i \in \mathcal{M}$, $\forall j \in \mathcal{N}$}
\end{equation}
The objective function (\ref{eq:2}) is about maximizing the profit by scheduling the deliveries with $m$ drones. Constraint (\ref{eq:3}) depicts that a drone cannot have an assignment that exceeds it's battery budget $B$. Constraint (\ref{eq:4}) says that a delivery can be completed by at most one drone. Constraint (\ref{eq:5}) tells that if two deliveries are in conflict, then a drone can choose at most one of them.

The aforementioned formulation is only suitable for solving the problem optimally for small-sized instances. But for large instances, it is not suitable. Hence, we propose following two approximation algorithms for solving the problem with large instances; one for a single drone and the other for multiple drones.
\vspace{-4mm}
\section{Delivery with Single Drone}
\label{section:singleDrone}
In this section, we design an approximation algorithm \textsc{(ApproxAlgoForSDSP)}
to solve the above problem for single drone case, i.e., $|\mathcal{M}| = 1$, which is referred as single drone-delivery scheduling problem (SDSP). 

According to the model, intervals $I_{1}, I_{2}, \cdots, I_{n}$ are sorted in non-decreasing order as per their launching time. Now, we can find the previous nearest conflict free interval $I_{L(j)}$ for each $I_{j} \in I$ by computing $L(j)$ as follows.

$L(j) := \max \{k : k \in \mathcal{N}$, $k < j$ and $t_{k}^{R} < t_{j}^{L}$\}, $\forall j \in \mathcal{N}$ and $L(j) = 0$ if there is no such $k$ exists for the delivery $j$.

This procedure gives us knowledge about which intervals we need to pull out from the previous assignment if we add the current interval to it.

Now, we design a pseudo polynomial dynamic programming (Algorithm \ref{Alg:EASDSP})  for SDSP and then an approximation algorithm (Algorithm \ref{Alg:AASDSP}) based on the FPTAS for the \textit{0-1 knapsack problem} \cite{Vazirani2003}.
\subsection{Dynamic Programming Formulation}
Let $P =\max_{j \in \mathcal{N}} (p_{j})$, then $nP$ is an upper bound of the optimal profit.
We define $S(j,p) \subseteq \{I_{1}, I_{2}, \cdots, I_{j}\}$,  a compatible subset of $I$ such that the cost $\mathcal{W}(S(j,p))$ is minimized subject to $\mathcal{P}(S(j,p)) = p$, where $1 \leq j \leq n$ and $0 \leq p \leq nP$. \\
Set $A(j,p) = \mathcal{W}(S(j,p))$.\\
\textbf{Basis}:- $A(1,p)$ $=$ $c_{1}$, if $p$ $=$ $p_{1}$; otherwise $A(1,p) = \infty$.\\
\phantom{x} \hspace{1cm} $A(j,0) = 0$; $\forall$ $j$ $\in$ $\mathcal{N}$ $\cup$ $\{0\}$.\\
\phantom{x} \hspace{1cm} $A(0,p) =$ $\infty$; $\forall p: 1 \leq p \leq nP$.\\
For a given $A(j,p)$ $( 1 \leq j < n$ and $0 \leq p \leq nP )$, we can compute $A(j+1,p)$ as follows.
\begin{description}
\item{\bf Case 1 ($I_{j+1} \in S(j+1,p)$):}
All the deliveries between $(L(j+1) + 1)$ and $j$ do not be the part of $S(j+1,p)$. So, for this case $A(j+1, p)$ $=$ $c_{j+1}$ + $A(L(j+1) ,p-p_{j+1}$), if $p$ $\geq$ $p_{j+1}$; else $A(j+1,p) = A(j,p)$.
\item{\bf Case 2 ($I_{j+1}$ $\notin$ $S(j+1,p)$):}
For this case we always have $A(j+1,p)$ $=$ $A(j,p)$.
\end{description}
From the above two cases, we can infer that
\begin{equation}
	\begin{split}
		A(j+1,p) & = \min \{A(j,p), (c_{j+1} + A(L(j+1),p-p_{j+1}))\}, \text{  if p $\geq$ $p_{j+1}$},\\
		& = A(j,p), \text{ otherwise}.
		\label{eq:7}
	\end{split}
\end{equation}
Using the above equation (\ref{eq:7}), we design an optimal dynamic programming algorithm (Algorithm \ref{Alg:EASDSP} (\textsc{DPAlgoForSDSP})) as follows. 
The optimal feasible intervals $\mathcal{S}^{*}$ $=$ \{$S^{*}$\} is an assignment for the drone to be delivered to make the profit maximum. 
\vspace{-4mm}
\begin{algorithm}[H]
\caption{(\textsc{DPAlgoForSDSP})}\label{Alg:EASDSP}
\begin{algorithmic}[1]
\For{$j\gets 1, n-1$}
\For{$p\gets 1, nP$}
\State compute $A(j+1,p)$ using the formula at the equation (\ref{eq:7}).
\EndFor
\EndFor
\State Find $p$: $p = \max \{p : A(n,p) \leq B\}$ \Comment{Maximum profit $p$ subject to the budget $B$}.
\State Find the feasible intervals ($S^{*}$) by backtracking to the array A for the value $p$.
\end{algorithmic}
\end{algorithm}
\vspace{-8mm}
\begin{theorem}
\label{Theorem-1}
Algorithm \ref{Alg:EASDSP} is a pseudo-polynomial time algorithm for SDSP.
\end{theorem}
\begin{proof}
  We can find the $L(j)$ for each $j$ $\in$ $\mathcal{N}$ in $\mathcal{O}(n \log n)$ time. For Algorithm \ref{Alg:EASDSP} we can compute the 2D array $(A)$ by $\mathcal{O}$($n^{2}$P) time. Also, step 6 of this algorithm can be done by $\mathcal{O}(nP)$ time. 
So, overall running time is $\mathcal{O}$($n^{2}P$) and correctness follows from the equation (\ref{eq:7}).\\ 
Thus this is a pseudo-polynomial time algorithm.  \qed
\end{proof}

Now, we propose a fully polynomial time approximation scheme (FPTAS) which is described via Algorithm  \ref{Alg:AASDSP} (\textsc{ApproxAlgoForSDSP}) to solve SDSP. Algorithm \ref{Alg:AASDSP} uses Algorithm \ref{Alg:EASDSP} as a subroutine.

\vspace{-4mm}
\begin{algorithm}
\caption{(\textsc{ApproxAlgoForSDSP})}\label{Alg:AASDSP}
\begin{algorithmic}[1]
\State Set $K$ $=$ $\frac{\epsilon P}{n}$, for $\epsilon$ $> 0$.
\State Set a new profit $p'_{j} = \lfloor \frac{p_{j}}{K} \rfloor$ for each delivery $j \in \mathcal{N}$.
\State Find the optimum solution ($S'$) of the same problem (SDSP) with the new profit $p'_{j}$ ($j \in \mathcal{N}$), using Algorithm \ref{Alg:EASDSP}.
\end{algorithmic}
\end{algorithm}
\vspace{-4mm}
\begin{theorem}
\label{Theorem-2}
Algorithm \ref{Alg:AASDSP} gives an FPTAS for SDSP with running time $\mathcal{O}$($n^{2}  \lfloor \frac{n}{\epsilon} \rfloor )$ for $\epsilon$ > 0.
\end{theorem}
\begin{proof}
Let $S^{*}$ and $S'$ be the optimum feasible intervals
for the SDSP with profit values $p_{j}$ and  $p'_{j}$ ($\forall j \in \mathcal{N})$, respectively. 

Also let $\mathcal{P} (S^{*})$ $=$ $\sum_{j \in S^{*}}^{}$ $p_{j}$ ; $\mathcal{P}' (S^{*})$ $=$ $\sum_{j \in S^{*}}^{}$ $p'_{j}$ ; $\mathcal{P}' (S')$ $=$ $\sum_{j \in S'}$ $p'_{j}$.\\
As ($p_{j}$ - $K.p'_{j}$) $\leq$ $K$, we have $\mathcal{P}(S^{*})$ - $K \mathcal{P}'(S^{*})$ $\leq$ $nK$.\\
As Algorithm \ref{Alg:AASDSP} gives us the optimum solution for the problem with new profit values and $S'$, $S^{*}$ both are being feasible sets, $\mathcal{P}' (S')$ $\geq$ $\mathcal{P}' (S^{*})$.\\
Also we have $\mathcal{P} (S') \geq K.\mathcal{P}' (S')$, as $p_{j} \geq K.p'_{j}$.
Thus,
\begin{equation}
\begin{split}
		\mathcal{P} (S') & \geq K.\mathcal{P}' (S')	\geq K.\mathcal{P}' (S^{*})
		 \geq \mathcal{P} (S^{*}) - nK
		 = \mathcal{P} (S^{*}) - \epsilon P\\
		& \geq \mathcal{P} (S^{*}) - \epsilon \mathcal{P} (S^{*}), \text{ as $\mathcal{P}$ ($S^{*}$) $\geq$ P}\\
	&	 = (1 - \epsilon) \mathcal{P} (S^{*})
\end{split}
\end{equation}
The running time for this algorithm is $\mathcal{O}$($n^{2}  \lfloor
\frac{P}{K} \rfloor )$ $=$ $\mathcal{O}$($n^{2}  \lfloor
\frac{n}{\epsilon} \rfloor )$ follows from Theorem \ref{Theorem-1}.
Hence, Algorithm \ref{Alg:AASDSP} is an FPTAS  for the SDSP.
\qed
\end{proof}

\section{Delivery with Multiple Drones}
\label{section:multiDrone}
In this section, we formulate an approximation algorithm for the  \textit{multiple drone delivery scheduling problem} (MDSP). 
We assume that for each delivery $j \in \mathcal{N}$ associate cost, $c_{j} \leq B$, otherwise it will not be part of any feasible solution.
Please note that, unlike SDSP, MDSP does not have an FPTAS, as stated below.

\begin{theorem}
MDSP does not admit an FPTAS. 
\end{theorem}
\begin{proof}
MDSP is a generalized version of \textit{multiple knapsack problem} (MKP). If we remove the conflicting condition of MDSP, it becomes MKP. Chekuri et al. \cite{fptas} proved that MKP does not admit FPTAS. From this, it is straightforward to conclude that  MDSP also does not admit FPTAS. \qed 
\end{proof}
In the following part, we propose a constant factor approximation algorithm to solve the MDSP with a restriction on the number of drones.
Now, we define some definitions and lemmata for the MDSP.

\begin{definition}
    The \textit{density} of a delivery $j \in \mathcal{N}$ defined as the ratio between profit ($p_j$) and cost ($c_j$), which is denoted by $d_{j}$ $=$ $\frac{p_{j}}{c_{j}}$.
\end{definition}
\begin{definition}
An assignment $S = \{I_{k_1}, I_{k_2}, \cdots, I_{k_l}\} \subseteq I$ of the deliveries with density $d_{k_1} \geq d_{k_2} \geq \cdots \geq d_{k_l}$, is called a \textit{critical assignment} for a drone with battery budget $B$, if $\mathcal{W}(S) > B$ and $\mathcal{W}(S \setminus \{k_{l}\}) \leq B$.
\end{definition}
We can state the following lemma based on \cite{knapsack}.

\begin{lemma}
\label{lemma-1}
If $S = \{I_{k_1}, I_{k_2}, \cdots, I_{k_l}\} \subseteq I$ with $l$ largest densities $d_{k_1} \geq d_{k_2} \geq \cdots \geq d_{k_l}$ is a compatible and critical assignment for any SDSP  then $\mathcal{P}(S) \geq f(1).$
\end{lemma}
\begin{proof}
We prove the lemma by contradiction.
In contradictory, let $\mathcal{P}(S) < f(1)$.\\
Let $T$ be the optimal assignment for the chosen SDSP, i.e., $\mathcal{P}(T) = f(1)$.\\
Then $\mathcal{P}(S) < \mathcal{P}(T)$ implies $\mathcal{P}(S \setminus T) < \mathcal{P}(T \setminus S)$.\\
But we have $d(S \setminus T) \geq d(T \setminus S)$, implies $\mathcal{W}(S \setminus T) < \mathcal{W}(T \setminus S)$, implies $\mathcal{W}(S) < \mathcal{W}(T)$, which is a contradiction as $\mathcal{W}(T) \leq B$ and $\mathcal{W}(S) > B$.\\
Hence, the statement of the lemma follows. \qed
\end{proof}
From the above lemma, we can state the following result.
\begin{corollary}
\label{corollary-1}
 $\max(\mathcal{P}(\{I_{k_{l}}\}), \mathcal{P}(S \setminus \{I_{k_{l}}\})) \geq \frac{f(1)}{2}$.
\end{corollary}
\begin{lemma}
If $\mathcal{S} = \{S_{1}, S_{2}, \cdots, S_{m}\}$, with first $\sum_{i=1}^{m} |S_{i}|$ many largest density intervals, is a collection of assignment for MDSP with $m$ many drones, where $S_i$ is a compatible and critical assignment for the drone $i$ $(1 \leq i \leq m)$ and  $S_k \cap S_l = \emptyset$; $\forall k, l: 1 \leq k \neq l \leq m$, then $\mathcal{P}(\mathcal{S}) = \sum_{i=1}^{m} \mathcal{P}(S_{i}) \geq f(m)$ holds.
\end{lemma}
\begin{proof}
Similar to the lemma \ref{lemma-1}. \qed
\end{proof}
\begin{theorem}
\label{Theorem-4}
Let, $\mathcal{S} = \{S_{1}, S_{2}, \cdots, S_{m+\Delta}\}$ be a collection of assignment for MDSP with ($m + \Delta$) many drones, where $S_i$ is a compatible assignment for the drone $i$ $(1 \leq i \leq m + \Delta)$ and $S_k \cap S_l = \emptyset$; $\forall k, l: 1 \leq k \neq l \leq m + \Delta$. If $\mathcal{S}$ is the set of first $\sum_{i=1}^{m + \Delta} |S_{i}|$ many largest density interval of $I$ and among  ($m + \Delta$) assignments of $\mathcal{S}$ at least $m$ are critical, then $\mathcal{P}(\mathcal{S}) = \sum_{i=1}^{m + \Delta} \mathcal{P}(S_{i}) \geq f(m)$ holds.

\end{theorem}
\begin{proof}
We prove the theorem by contradiction.\\
In contradictory, let $\mathcal{P}(\mathcal{S}) < f(m)$. Let $\mathcal{T}$ be the optimal family of assignment for the MDSP with $m$ drones, then $\mathcal{P}(\mathcal{T}) = f(m)$.

Let $S = \cup_{i=1}^{m+\Delta} S_i$ and $T$ be the set of all deliveries in $\mathcal{T}$. 
Then $\mathcal{P}(S) < \mathcal{P}(T)$, implies $\mathcal{P}(S \setminus T) < \mathcal{P}(T \setminus S)$. But we have $d(S \setminus T) \geq d(T \setminus S)$. So,
\begin{alignat}{2}
  &d(S \setminus T)
  \geq d(T \setminus S) &\& \hspace{16pt} \mathcal{P}(S \setminus T) < \mathcal{P}(T \setminus S)\\
  &\Rightarrow \mathcal{W}(S \setminus T) < \mathcal{W}(T \setminus S)\\
  &\Rightarrow \mathcal{W}(S) < \mathcal{W}(T) \label{eq:11}
\end{alignat}
The inequality (\ref{eq:11}) is a contradiction as $\mathcal{W}(T) \leq mB$ and $\mathcal{W}(S) > mB$.\\
Hence, the statement of the theorem follows.\qed
\end{proof}
We can construct an interval graph using the delivery intervals for a given instance of MDSP and find the maximum degree $\Delta$ of the graph. With this knowledge of $\Delta$, we propose an approximation Algorithm \ref{Algorithm-4} (\textsc{GreedyAlgoForMDSP}) to solve MDSP.
\vspace{-6mm}
\begin{algorithm}[H]
\caption{(\textsc{GreedyAlgoForMDSP})}\label{Algorithm-4}
\begin{algorithmic}[1]
\State \textbf{Initially:-} $M$ $=$ $\{1, 2, \cdots , m + \Delta\}$; $M'$ $=$ $\emptyset$; $\forall$ $i$ $\in$ $M:$ $L_{i}$ $=$ $\emptyset$;
$P_{i}$ $=$ $0$; $W_{i}$ $= 0$; $S_{i}$ $=$ $\emptyset$.
\State Sort delivery intervals $I$ according to their density. \Comment{wlog let $d_{1} \geq d_{2} \geq, \cdots, \geq d_{n}$.}
\For{$j\gets 1, n$}
\If{|$M'$| $\neq m$}
\State Find $i$ $\in$ M such that $S_{i} \cap I_{j} = \emptyset$
\If{$W_{i} + c_{j} \leq B$}
\State $W_{i} = W_{i} + c_{j}$; $S_{i} = S_{i} \cup \{{I_{j}}\}$; $P_{i} = P_{i} + p_{j}$.
\Else
\State $W_{i} = W_{i} + c_{j}$; $S_{i} = S_{i} \cup \{{I_{j}}\}$; $P_{i} = P_{i} + p_{j}$
\State $M = M \setminus \{i\}; M' = M' \cup \{i\}; L_{i} = \{j\}$.
\EndIf
\Else
\State break.
\EndIf
\EndFor
\For{each $i$ in $M'$}
\If{($P_{i} - \mathcal{P}(L_{i})) \geq \mathcal{P}(L_{i})$}
\State $W_{i} = W_{i} - \mathcal{W}(L_{i})$; $P_{i} = P_{i} - \mathcal{P}(L_{i})$; $S_{i} = S_{i} \setminus L_{i}$
\Else
\State $W_{i} = \mathcal{W}(L_{i})$; $P_{i} = \mathcal{P}(L_{i})$; $S_{i} = L_{i}$.
\EndIf
\EndFor
\State Return $m$ most profitable $i$ $\in$ $(M \cup M')$ along with it's $S_i$.
\end{algorithmic}
\end{algorithm}
\vspace{-8mm}
\begin{lemma}
\label{lemma-3}
Algorithm \ref{Algorithm-4} gives a feasible solution for MDSP with $m$ drones.
\end{lemma}
\begin{proof}
Instead of $m$ drones, the algorithm starts with ($m + \Delta$) drones. For assigning a delivery $j$ (in order to their non-decreasing order of density) to a drone in $M$ at step $3-15$, the algorithm checks whether there are $m$ critical assignments (recorded by $M'$) in the current schedule or not. If not (i.e., $M' \neq m$), then there is at least ($\Delta + 1$) drones in $M$, because $|M|\geq (m + \Delta) - (m - 1)$, to assign the delivery $j$. In this schedule, $j$ can conflict with the assignment of at most $\Delta$ many drones in $M$. Therefore, there always exists a drone $i$ in $M$ such that it's current assignment is compatible with $I_j$, i.e., $S_{i} \cap I_{j} = \emptyset$. Now, if the addition of delivery $j$ to $S_i$ does not exceed the battery budget, i.e., $W_{i} + c_{j} \leq B$, where $W_i = \mathcal{W}(S_i)$. Then drone $i$ is feasible for the delivery $j$. So, in this case, we schedule the delivery $j$ to the drone $i$ and update the current cost ($W_i$) and profit ($P_i$) value for the drone $i$ at step $7$. Otherwise ($W_{i} + c_{j} > B$), the addition of delivery $j$ to the current schedule ($S_i$) of $i$ makes the assignment critical. So, in this case we also update the current assignment ($S_i$), cost ($W_i$) and profit ($P_i$) value for the drone $i$ at step $9$. In addition, we remove the drone $i$ from the set $M$ and add it to the set $M'$ at step $10$.

So, for any delivery, there is always a drone whose present assigned deliveries are not conflicting with the considered one. That's why either we can find $m$ critical assignments, whose indices are recorded by $M'$, or all the deliveries are assigned to the $(m + \Delta)$ drones. Whenever we find $m$ critical assignment, i.e., $|M'| = m$, we break the loop at step $13$.

At the end of step $15$, we have $(m + \Delta)$ compatible assignments, among them at most $m$ are critical. To make those assignments feasible, for each assignment, we ignore either the last delivery or all the preceding deliveries for which the profit is maximum.

But our objective is to find $m$ feasible assignments. For that at step 23, we take $m$ most profitable assignments among all $(m + \Delta)$ feasible assignments.
Thus Algorithm \ref{Algorithm-4} finds a feasible solution for MDSP with $m$  drones. \qed
\end{proof}

\begin{lemma}\label{lemma-4}
The time complexity of  
Algorithm \ref{Algorithm-4} is $\mathcal{O}$($n^{2}$), where $n= |\mathcal{N}|$. 
\end{lemma}
\begin{proof}
Step-wise analysis of the Algorithm \ref{Algorithm-4} is given below.
The $n$ delivery intervals can be sorted at step 2 in $\mathcal{O}(n \log n)$ time.
For assigning the delivery $j$ according to their non-decreasing order of density to a drone in $M$, at most $(\Delta + 1)$ compatibility check to any of the ($\Delta+1$) assignments of the drones in $M$ is sufficient. Since there are at most ($j-1$) deliveries assigned to those drones, we can find one such drone $i$ in $M$ for which $S_{i} \cap I_{j} = \emptyset$ by $\mathcal{O}(j)$ many comparisons. Then all the update operations corresponding to $S_i$, $W_i$, $P_i$, $M$ and $M'$ take constant time. 
Hence, total running time for step (3-15) is $\leq$ $\sum_{j=1}^{n}$  $\mathcal{O}(j)$ $=$ $\mathcal{O}(n^{2})$.

Step (16-22) runs in $\leq$ $\mathcal{O}(|M^{'}|)$  $\leq$ $\mathcal{O}(m)$ time. Last step can be done in $\mathcal{O}(m + \Delta) + \mathcal{O}(m \log (m + \Delta)) \leq  \mathcal{O}(n^{2})$  time. (as $m, \Delta \leq n$)\\
Thus, the overall running time is $\mathcal{O}(n^{2}$). \qed
\end{proof}
\begin{lemma}
\label{lemma-5}
Algorithm \ref{Algorithm-4} gives an $\frac{m}{2(m+\Delta)}$ approximation factor of the MDSP with $m$ many drones and $\delta$ being the maximum degree of the interval graph.
\end{lemma}
\begin{proof}
Let $OPT = f(m)$, be the optimum profit for the MDSP with $m$ many drones.
Let $\mathcal{S} = \{S_{1}, S_{2}, \cdots, S_{m+\Delta}\}$ be the family of assignments return by the Algorithm \ref{Algorithm-4} at the step 15. At this step, either all the deliveries are assigned to the $(m+\Delta)$ drones or $m$ assignments are critical. For any case $\mathcal{P}(\mathcal{S}) = \sum_{i=1}^{m + \Delta} \mathcal{P}(S_{i}) \geq OPT$.  ($m$ assignments are critical implies $\mathcal{P}(\mathcal{S}) \geq OPT$ from Theorem \ref{Theorem-4}).

Without loss of generality, let  \{$S_{1}, S_{2}, \cdots, S_{t}$\} $(t \leq m)$ be a collection of critical assignment. To make these assignments feasible, the algorithm moderates each by omitting either the last delivery or all the remaining deliveries for which the profit is maximum.
Let \{$S'_{1}, S'_{2}, \cdots, S'_{t}$\} be the moderate set. Then $\mathcal{P}(S^{'}_{i}) \geq \frac{\mathcal{P}(S_{i})}{2}$, $\forall i: 1 \leq i \leq t$ (from corollary \ref{corollary-1}).

The final solution returns by the algorithm is those assignments corresponding to the $m$ largest profits among $\{\mathcal{P}(S_{1}^{'}), \mathcal{P}(S_{2}^{'}), \cdots, \mathcal{P}(S_{t}^{'}), \mathcal{P}(S_{t+1}), \mathcal{P}(S_{t+2}), $ $\cdots, \mathcal{P}(S_{m+\Delta})\}$,
let $\mathcal{S}_{A}$ be that collection. Then,
\begin{equation}
	\begin{split}
		\mathcal{P} (\mathcal{S}_{A}) & \geq \frac{m}{m+\Delta} ( \sum_{i=1}^{t} \mathcal{P}(S'_{i}) + \sum_{i = t+1}^{m+\Delta} \mathcal{P}(S_{i}) )\\
		& \geq \frac{m}{m+\Delta} ( \sum_{i=1}^{t} \frac{1}{2} \mathcal{P}(S_{i}) + \sum_{i = t+1}^{m+\Delta} \mathcal{P}(S_{i}) )
		 \geq \frac{m}{m+\Delta} .\frac{1}{2} ( \sum_{i=1}^{m+\Delta} \mathcal{P}(S_{i}) )\\ 
	&	 = \frac{m}{2(m+\Delta)} \mathcal{P}(\mathcal{S})
		 \geq \frac{m}{2(m+\Delta)} OPT.
	\end{split} 
\end{equation} 
\vspace{-3.5mm}
\qed
\end{proof}
\begin{theorem}
For $m \geq \Delta$ Algorithm \ref{Algorithm-4} gives an $\frac{1}{4}$-approximation algorithm for MDSP, where $m$ is the number of drones and $\Delta$ is the maximum degree of the interval graph constructed from given delivery time intervals.
\end{theorem}
\begin{proof}
Correctness of the algorithm follows from lemma \ref{lemma-3} and polynomial running time follows from lemma \ref{lemma-4}.

Now, for $m$ $\geq$ $\Delta$, $\frac{m}{m+\Delta}$ $\geq$ $\frac{1}{2}$, implies $\frac{m}{2(m+\Delta)}$ $\geq$ $\frac{1}{4}$. 
Hence from lemma \ref{lemma-5}, if $\mathcal{S}_{A}$ is the solution of MDSP, derived by using the algorithm \ref{Algorithm-4} and OPT is the optimum profit of that problem, then ${\mathcal{P}} (\mathcal{S}_{A}) \geq \frac{1}{4}OPT$.
Hence the proof. \qed
\end{proof}
\section{Conclusion}
\label{section:conclusion}
This paper studied the drone-delivery scheduling problem for single and multiple drone cases. We propose an FPTAS for the single drone-delivery scheduling problem (SDSP) with running time $\mathcal{O}$($n^{2}  \lfloor \frac{n}{\epsilon} \rfloor )$ for $\epsilon$ > 0 and $n$ being the number of deliveries. Also propose an $\frac{1}{4}$-approximation algorithm with running time $\mathcal{O}(n^{2})$ for MDSP with $m$ drones if $m$ $\geq$ $\Delta$, where $\Delta$ is maximum degree of the interval graph constructed from given delivery time intervals.
It would be interesting to find a constant factor approximation algorithm for future work rather than an asymptotic PTAS for MDSP without any constraints on the number of drones. Also, if drones are not identical concerning the battery budget, what would be the approach for an approximation algorithm?
\bibliographystyle{splncs03}
\bibliography{bib-nourl}

\centering
\begin{table}[H]
\setlength{\tabcolsep}{10pt}
\caption{
List of variables used in this paper}
	\begin{tabular}{|p{0.3cm}|p{5cm}|p{0.3cm}|p{5.87cm}|}
		\hline
		$\mathcal{N}$& Set of deliveries& $L(j)$& Index of nearest conflict-free interval of $I_j$. \\
		\hline
 		$\mathcal{M}$& Set of drones& $\mathcal{P}$& Profit function.  \\
		\hline
		$n$ & No. of deliveries& $\mathcal{W}$& Cost function.  \\
		\hline
		$m$ & No. of drones& $B$& Battery budget for a drone. \\
		\hline
		$i$ & Index for a drone& $L_{i}$& Last delivery assigned to the drone $i$.  \\
		\hline
 		$j$ & Index for a delivery & $M^{'}$& Set of drones with critical assignments. \\
 		\hline
		$\delta_{j}$& Location for delivery $j$& $P$& Maximum profit. \\
 		\hline
		$\delta_{j}^{L}$& Launch point for  delivery $j$& $S$& Set of delivery time intervals.  \\
 		\hline
		$\delta_{j}^{R}$& Rendezvous point for delivery $j$& $P_{i}$& Profit for the drone $i$.  \\
 		\hline
 		$t_{j}^{L}$& Launch time for delivery $j$& $W_{i}$& Cost for the the  drone $i$. \\
 		\hline
		$t_{j}^{R}$& Rendezvous time for delivery $j$& $S_{i}$& Assignment for the the drone $i$.  \\
 		\hline
 		$d_{j}$& Density/profit per cost of delivery $j$& $\Delta$& Maximum degree of the interval graph. \\
		\hline
		\hline
	\end{tabular}

\label{tab:notation}		
\end{table}
\end{document}